\documentclass[12pt,a4,portrait]{article}

\usepackage{hyperref}
\usepackage{bigints}
\usepackage{eufrak}
\usepackage{amsthm}
\usepackage{mathrsfs}
\usepackage{tocbibind}

 \newtheorem{thm}{Theorem}[section]
 \newtheorem{cor}[thm]{Corollary}
 \newtheorem{lem}[thm]{Lemma}
 \newtheorem{prop}[thm]{Proposition}
 \theoremstyle{definition}
 \newtheorem{defn}[thm]{Definition}
 \theoremstyle{remark}
 \newtheorem{rem}[thm]{Remark}
 \newtheorem*{ex}{Example}
 \numberwithin{equation}{section}

\begin{document}
	
\title{\bf Some remarks concerning invariant quantities in scalar-tensor gravity}
\author{Ott Vilson\thanks{ovilson@ut.ee} \\
{\normalsize Institute of Physics, University of Tartu,} \\ {\normalsize Ravila 14c, Tartu 50411, Estonia}}	

\date{}
\maketitle

\begin{abstract}
The aim of the current paper is to clarify some aspects of the formalism used for describing the scalar-tensor gravity characterized by four arbitrary local functionals of the scalar field. We recall the objects that are invariant with respect to a spacetime point under the local Weyl rescaling of the metric and under the scalar field redefinition. We phrase and prove a theorem that allows to link such an object to each quantity in a theory where two out of the four arbitrary local functionals of the scalar field are specified in a suitable manner. Based on these results we phrase and reason the existence of the so called translation rules. 

\vspace{0.5cm}

\end{abstract}

\maketitle
\section{Introduction}
The history of scalar-tensor theories of gravity (STG) is long, starting with the works of Jordan \cite{jordan} and Fierz \cite{Fierz}, later developed by Brans and Dicke \cite{BD}, \cite{Dicke}. The original idea was purely theoretical since there were no observational contradictions to Einstein's general relativity (GR). In about a decade ago astronomers claimed that the Universe is expanding in an accelerating manner and explained that in the context of GR with a nonvanishing cosmological constant. This needs finetuning which we would like to avoid in a fundamental theory. Due to the latter studying the extensions of GR, STG being one of them, is still popular. 

The aim of the current paper is to clarify some mathematical issues concerning the invariant quantities in general STG and the so called translation rules that were proposed in our recent paper \cite{JKSV_2}. A more detailed introduction and references to the literature on that subject can also be found there.

The outline of the paper is the following. In Section~\ref{parametrizations} we recall the general framework for STG mostly relying on the paper by Flanagan \cite{Flanagan}. Section~\ref{Invariants} summarizes the results of Ref.~\cite{JKSV_2} that will be used in the current paper. In Section~\ref{invariants_and_parametrizations} we phrase and prove a lemma and a theorem claiming the existence of the so called invariant pair. In Section~\ref{translation_rules} we point out an important corollary of the latter. Based on these results we formulate and reason the existence of the so called translation rules proposed in Ref.~\cite{JKSV_2}.  

\section{Parametrizations in scalar-tensor theories of gravity}\label{parametrizations}

In a scalar-tensor theory of gravity the gravitational interaction is characterized by a metric tensor $g_{\mu\nu}(x^\mu)$ of a curved spacetime $x^\mu \in V_4$ and a scalar field $\Phi(x^\mu)$. In the current paper we consider a family of scalar-tensor theories of gravity by postulating a general action functional \cite{Flanagan}
\begin{align}
\nonumber
S &= \frac{1}{2\kappa^2}\int_{V_4}d^4x\sqrt{-g}\left\lbrace {\mathcal A}(\Phi)R-
{\mathcal B}(\Phi)g^{\mu\nu}\nabla_\mu\Phi \nabla_\nu\Phi - 2\ell^{-2}{\mathcal V}(\Phi)\right\rbrace \\
\label{fl_moju}
&\quad\quad+ S_m\left[e^{2\alpha(\Phi)}g_{\mu\nu},\chi\right] \,
\end{align}
which contains four arbitrary local functionals $\left\lbrace \mathcal{A}(\Phi),\, \mathcal{B}(\Phi),\, \mathcal{V}(\Phi),\, \alpha(\Phi) \right\rbrace$ of the dimensionless scalar field $\Phi(x^\mu)$. Out of the four the local functional $\mathcal{A}(\Phi)$ is multiplied by the Ricci scalar $R$ and occasionally the term `curvature coupling' is used to refer to $\mathcal{A}(\Phi)$. Analogically `kinetic coupling' refers to $\mathcal{B}(\Phi)$, i.e.\ to the multiplier of the kinetic term for the scalar field $\Phi(x^\mu)$. The local functional $\mathcal{V}(\Phi)$ is known as the scalar field potential and from the particle physics viewpoint it contains the scalar field self-interactions. For a general case the matter action functional $S_m$ depends on the metric tensor $g_{\mu\nu}$ via conformal coupling $e^{2\alpha(\Phi)}$, i.e.\ the spacetime indexes in the Lagrangian for the matter fields, collectively denoted as $\chi$, are contracted by $e^{2\alpha(\Phi)}g_{\mu\nu}$ and its inverse. The term `matter coupling' is frequently used to refer to $\alpha(\Phi)$. Due to suitably chosen dimensionful constants $\kappa^2$ and $\ell^{-2}$ the four arbitrary local functionals $\left\lbrace \mathcal{A}(\Phi),\, \mathcal{B}(\Phi),\, \mathcal{V}(\Phi),\, \alpha(\Phi) \right\rbrace$ are dimensionless and if the functional form w.r.t.\ $\Phi(x^\mu)$ of each of them is fixed then the theory is fixed. Let us point out that all local functionals of $\Phi(x^\mu)$ inherit a dependence on $x^\mu$ and hence are functions of a spacetime point as well.

\begin{prop}\label{prop:fl_fn_teisenemine}
If under the local Weyl rescaling of the metric tensor and under the scalar field redefinition
\begin{align}
	\label{conformal_transformation}
	g_{\mu\nu} &= e^{2\bar{\gamma}(\bar{\Phi})}\bar{g}_{\mu\nu} \,,	 \\
	\label{field_redefinition}
	\Phi &= \bar{f}(\bar{\Phi}) \,
\end{align}
the four arbitrary local functionals are imposed to transform as
\begin{subequations}
	\label{fl_fn_teisendused}
\begin{align}
	\label{A_transformation}
		{\mathcal A} \left( {\bar f}( {\bar \Phi})\right) &= e^{-2\bar{\gamma}(\bar{\Phi})} \bar{\mathcal{A}}(\bar{\Phi})
		 \,,\\
	\label{B_transformation}
		{\mathcal B}\left(\bar{f}(\bar{\Phi})\right) &= e^{-2{\bar \gamma}({\bar \Phi})} \left(\bar{f}^\prime\right)^{-2} \left( {\bar {\mathcal B}}({\bar \Phi}) - 6\left(\bar{\gamma}^{\,\prime}\right)^2 \bar{\mathcal{A}} \left(\bar{\Phi} \right) +
		6\bar{\gamma}^{\,\prime} \bar{\mathcal{A}}^\prime \right) \,, \\
	\label{V_transformation}
		{\mathcal V}\left(\bar{f}(\bar{\Phi})\right) &= e^{-4\bar{\gamma}(\bar{\Phi})} \, \bar{{\mathcal V}}(\bar{\Phi}) \,, \\
	\label{alpha_transformation}
		\alpha\left(\bar{f}(\bar{\Phi})\right) &= \bar{\alpha}(\bar{\Phi}) - \bar{\gamma}(\bar{\Phi})\,
\end{align}
\end{subequations}
then the action functional \eqref{fl_moju} is invariant under the transformations \eqref{conformal_transformation}-\eqref{field_redefinition} up to a boundary term \cite{Flanagan}. 
\end{prop}
\noindent Here and in the following we shall drop the arguments of the functionals unless confusion might arise. Let us also adopt a notation where prime as a superscript of a ``barred" local functional of the scalar field means variational derivative w.r.t.\ the ``barred" scalar field $\bar{\Phi}(x^\mu)$ and prime as a superscript of such a quantity without ``bar" means variational derivative w.r.t.\ the ``unbarred" scalar field $\Phi(x^\mu)$, e.g.\ $\bar{f}^\prime  \equiv \displaystyle{\frac{\delta \bar{f}(\bar{\Phi})}{\delta \bar{\Phi}}}$ and $\mathcal{A}^\prime \equiv \displaystyle{\frac{\delta \mathcal{A}(\Phi)}{\delta \Phi}}$ respectively. Note that due to the inherited dependence on a spacetime point one can differentiate functionals of $\Phi$ w.r.t.\ $x^\mu$ via ordinary partial derivatives.

The relations \eqref{fl_fn_teisendused} are obtained by rewriting the action functional \eqref{fl_moju} using $\bar{g}_{\mu\nu}$ and $\bar{\Phi}$ as dynamical fields. In the current paper we assume the affine connection to be the Levi-Civita one. Due to the latter such a rewriting of the action functional \eqref{fl_moju} also introduces a boundary term but here and in the following we shall drop boundary terms. We also assume the premiss of Proposition \ref{prop:fl_fn_teisenemine} to hold and whenever Eqs.~\eqref{conformal_transformation}-\eqref{field_redefinition} are recalled also Eqs.~\eqref{fl_fn_teisendused} are taken into account.
\begin{defn}[parametrization]\label{def:parametrization}
If the functional form w.r.t.\ $\Phi$ of exactly two out of the four arbitrary local functionals $\left\lbrace \mathcal{A},\, \mathcal{B},\, \mathcal{V},\, \alpha \right\rbrace$ is fixed then we say that the theory is given in a specific frame and parametrization.
\end{defn}
\noindent The term `reparametrization' refers to the scalar field redefinition \eqref{field_redefinition} while the Weyl rescaling \eqref{conformal_transformation} is the change of the `frame'. Roughly speaking both of these transformations can be used to fix the functional form of one arbitrary local functional out of the four. A closer look on the transformation properties \eqref{fl_fn_teisendused} reveals that all four arbitrary local functionals transform under the Weyl rescaling \eqref{conformal_transformation} but it might be the case that not all of them transform under the scalar field redefinition \eqref{field_redefinition} (e.g.\ $\mathcal{A}=1$). Therefore it is convenient to think that first the frame is chosen, i.e.\ we specify the metric tensor, and then the parametrization is chosen. In that sense the latter involves the former and in the following an explicit reference to the chosen frame is suppressed.
\begin{ex}
	The Jordan frame Brans-Dicke-Bergmann-Wagoner parametri\-zation (JF BDBW) with the scalar field denoted as $\Psi$ is given by \cite{BD}, \cite{Faraoni}, \cite{Maeda}: 
	\begin{equation}
	\label{Jordan_frame}
	\mathcal{A} \equiv \Psi \quad ,\quad \mathcal{B} \equiv \frac{\omega(\Psi)}{\Psi} \quad ,\quad \mathcal{V} \equiv \mathcal{V}_{\mathfrak{J}}(\Psi) \quad ,\quad \alpha \equiv 0 \,. 
	\end{equation}
	The Einstein frame canonical parametrization (EF canonical) with the scalar field denoted as $\varphi$ is given by \cite{Dicke}, \cite{Faraoni}, \cite{Maeda}:
	\begin{equation}
	\label{Einstein_frame}
	\mathcal{A} \equiv 1 \quad ,\quad \mathcal{B} \equiv 2 \quad ,\quad \mathcal{V} \equiv \mathcal{V}_{\mathfrak{E}}(\varphi) \quad ,\quad \alpha \equiv \alpha_{\mathfrak{E}}(\varphi) \,.
	\end{equation}
\end{ex}
A parametrization is in principle meaningful without considering the Weyl rescaling \eqref{conformal_transformation} and the scalar field redefinition \eqref{field_redefinition} at all but nevertheless in a generic case these transformations can be used to transform an arbitrary set of functionals $\left\lbrace \mathcal{A}(\Phi),\, \mathcal{B}(\Phi),\, \mathcal{V}(\Phi),\, \alpha(\Phi) \right\rbrace$ into e.g.\ JF BDBW parametrization \eqref{Jordan_frame}. Hence a chosen parametrization is not a unique description of a theory.
\begin{ex}
	In order to transform from JF BDBW parametrization \eqref{Jordan_frame} to EF canonical parametrization \eqref{Einstein_frame} we consider the relations
	\begin{equation}
	\label{EF_to_JF}
	e^{2\bar{\gamma}(\varphi)} = e^{2\alpha_{\mathfrak{E}}(\varphi)} \,,\quad \left(\frac{\delta \Psi}{\delta \varphi}\right)^2 = 4e^{-4\alpha_{\mathfrak{E}}(\varphi)}\left( \frac{\delta \alpha_{\mathfrak{E}}(\varphi)}{\delta \varphi} \right)^2 \rightarrow \Psi = \Psi(\varphi) 
	\end{equation}
	in the case when EF canonical parametrization quantities are considered to be the ``barred" ones. For the reverse transformation we choose
	\begin{equation}
	\label{JF_to_EF}
	e^{2\bar{\gamma}(\Psi)} = \Psi \,,\quad \left(\frac{\delta \varphi}{\delta \Psi}\right)^2 = \frac{2\omega(\Psi) + 3}{4\Psi^2} \rightarrow \varphi \equiv \varphi(\Psi) 
	\end{equation}
	if instead JF BDBW parametrization quantities are considered to be the ``barred" ones \cite{Faraoni}.
\end{ex} 

\section{Invariants}\label{Invariants}

Let us recall three basic objects introduced in our recent paper \cite{JKSV_2}
\begin{align}\label{I_1}
	\mathcal{I}_1(\Phi) &\equiv \frac{e^{2\alpha(\Phi)}}{\mathcal{A}(\Phi)} \,,\qquad 
	\mathcal{I}_2(\Phi) \equiv \frac{\mathcal{V}(\Phi)}{\mathcal{A}(\Phi)^2} \,,\, \\
	\label{I_3}
	\mathcal{I}_3(\Phi) &\equiv
	\pm \bigints \sqrt{ \frac{ 2\mathcal{A}(\Phi)\mathcal{B}(\Phi) + 3 \left( \mathcal{A}^\prime(\Phi)\right)^2}{4\mathcal{A}(\Phi)^2}} \delta \Phi \,.
\end{align}
In Eq.~\eqref{I_3} the integrand is a local functional of $\Phi$ but, as there is no dependence on the derivatives of $\Phi$, for such a case $\delta \Phi$ coincides with $d\Phi$ and the expression under consideration is in principle an ordinary indefinite integral.

Eqs.~\eqref{I_1}-\eqref{I_3} define functions of a spacetime point through three compositional steps: 
\begin{itemize}		
	\item[i)  ] $\mathcal{I}_i \equiv \mathcal{I}_i(\left\lbrace \mathcal{A},\,\mathcal{B},\,\mathcal{V},\,\alpha  \right\rbrace)$, e.g.\ $\mathcal{I}_1 \equiv \mathcal{I}_1(\mathcal{A},\,\alpha) \equiv \frac{e^{2\alpha}}{\mathcal{A}}$. \\
	The structure of $\mathcal{I}_i$ w.r.t.\ $\left\lbrace \mathcal{A},\,\mathcal{B},\,\mathcal{V},\,\alpha  \right\rbrace$ is preserved under the Weyl rescaling of the metric tensor \eqref{conformal_transformation} and the scalar field redefinition \eqref{field_redefinition}.
		
	\item[ii) ] $\mathcal{I}_i \equiv \mathcal{I}_i(\Phi) \Leftarrow \mathcal{A} \equiv \mathcal{A}(\Phi) \text{ etc}$. \\
	Under the Weyl rescaling $\mathcal{I}_i$ preserves its functional form w.r.t.\ the scalar field $\Phi$, i.e.\ $\bar{\mathcal{I}}_i(\bar{\Phi}) \equiv \mathcal{I}_i(\Phi\equiv \bar{\Phi})$. If also the scalar field $\Phi$ is redefined then $ \bar{\mathcal{I}}_i (\bar{\Phi}) \equiv \left( \mathcal{I}_i \circ \bar{f} \,\right) (\bar{\Phi})$.	
		
	\item[iii)] $\mathcal{I}_i \equiv \mathcal{I}_i(x^\mu) \Leftarrow \Phi \equiv \Phi(x^\mu)$.\\
	$\mathcal{I}_i$ is an invariant w.r.t.\ a spacetime point $x^\mu \in V_4$ which follows from the fact that under the transformations \eqref{conformal_transformation}-\eqref{field_redefinition} the numerical value of the four arbitrary local functionals at a spacetime point changes due to multiplicative and additive terms in Eqs.~\eqref{fl_fn_teisendused}. For $\mathcal{I}_i$ the extra terms and factors cancel out and hence the numerical value of $\mathcal{I}_i$ at a spacetime point is preserved under the transformations \eqref{conformal_transformation}-\eqref{field_redefinition}. In the same spirit we conclude that $\partial_\mu \mathcal{I}_i$ is also an invariant w.r.t.\ $x^\mu$.
\end{itemize}

\begin{cor}\label{cor1}
	One may define arbitrarily many quantities having the same transformation properties as $\mathcal{I}_1$ etc. via three procedures
	\begin{itemize}
		\item[i)  ] Introducing an arbitrary functional $h$
		\begin{equation}
		\mathcal{I}_i \equiv h\left( \left\lbrace \mathcal{I}_j \right\rbrace_{j\in \mathscr{J}} \right)
		\end{equation}
		where $\mathscr{J}$ is some set of indices.
		\item[ii) ] Introducing a quotient of derivatives
		\begin{equation}
		\mathcal{I}_j \equiv \frac{\mathcal{I}_i^\prime}{\mathcal{I}_k^\prime} \equiv ^{ \displaystyle{\frac{\delta \mathcal{I}_i}{\delta \Phi}} }\negmedspace \bigg/ \negmedspace _{ \displaystyle{\frac{\delta \mathcal{I}_k}{\delta \Phi}} } = \frac{\delta \mathcal{I}_i}{\delta \mathcal{I}_k}\,.
		\end{equation} 
		\item[iii) ] Integrating over the scalar field $\Phi$
		\begin{equation}
		\mathcal{I}_i \equiv \int \mathcal{I}_j \mathcal{I}_k^\prime \delta \Phi \,
		\end{equation}
		in the sense of an indefinite integral.
	\end{itemize}	
\end{cor}
\noindent We shall refer to such quantities as invariants.

\begin{ex}
	\begin{equation}\label{I_4}
	\mathcal{I}_4(\Phi) \equiv \frac{\mathcal{I}_2(\Phi)}{\mathcal{I}_1(\Phi)^2} \,,\qquad
	\mathcal{I}_5(\Phi) \equiv \left( \frac{ \mathcal{I}_1^\prime(\Phi) }{ 2\,\mathcal{I}_1(\Phi)\,\mathcal{I}_3^\prime(\Phi) } \right)^2 \,.
	\end{equation}	
\end{ex}

Let us introduce an `invariant metric' as
\begin{equation}\label{invariant_metric}
	\hat{g}^{(\cdot)}_{\mu\nu} \equiv \mathcal{I}_i \mathcal{A} g_{\mu\nu} \,.
\end{equation}
Here the precise definition depends on the choice of $\mathcal{I}_i$ and we shall distinguish between different invariant metrics by using some superscript $(\cdot)$. By Eq.~\eqref{invariant_metric} we have defined an object which under the Weyl rescaling of the metric tensor \eqref{conformal_transformation} and under the scalar field redefinition \eqref{field_redefinition} transforms as $\mathcal{I}_1$ etc. due to suitable transformation properties of $\mathcal{A}$ given by \eqref{A_transformation}. Nevertheless it is a metric tensor, e.g.\ it can be used to raise and lower spacetime indices.

We define the Levi-Civita connection with respect to $\hat{g}_{\mu\nu}^{(\cdot)}$ as
	\begin{align}
		\nonumber
		\hat{\Gamma}^{\sigma}_{\mu\nu} \equiv \Gamma^{\sigma}_{\mu\nu} + \frac{\mathcal{A}^\prime}{2\mathcal{A}}\left( \delta^\sigma_\mu \partial_\nu \Phi + \delta^\sigma_\nu \partial_\mu \Phi - g_{\mu\nu}g^{\sigma\rho} \partial_\rho\Phi \right) + \\
		\label{invariant_connection}
		+ \frac{1}{2\,\mathcal{I}_i} \left( \delta^\sigma_\mu \partial_\nu \mathcal{I}_i + \delta^\sigma_\nu \partial_\mu \mathcal{I}_i - g_{\mu\nu}g^{\sigma\rho} \partial_\rho\mathcal{I}_i \right)
	\end{align}
where $\Gamma^\sigma_{\mu\nu}$ are the Levi-Civita connection coefficients for the metric $g_{\mu\nu}$. 

\begin{rem}\label{rem3}
	The definition \eqref{invariant_connection} is in a sense identical to the well known transformation rule of the Levi-Civita connection coefficients under the Weyl rescaling of the metric tensor $g_{\mu\nu}$ \cite{Nakahara} but here the key idea is that we introduce additional terms to cancel the effect of the Weyl rescaling on $\Gamma^{\sigma}_{\mu\nu}$.
\end{rem}

The definitions \eqref{invariant_metric} and \eqref{invariant_connection} can be used to construct geometrical objects, such as $\hat{R}^{(\cdot)}$, that are invariant under the Weyl rescaling of the metric tensor \eqref{conformal_transformation}.

\section{Invariants and parametrizations}\label{invariants_and_parametrizations}

In what follows we shall work with three formulations of STG:
\begin{itemize}
	\item[i) ] The generic case described by the action functional \eqref{fl_moju} where non of the four arbitrary local functionals $\left\lbrace \mathcal{A},\,\mathcal{B},\,\mathcal{V},\,\alpha  \right\rbrace$ of $\Phi$ have gained a fixed functional form. We denote these variables as denoted in \eqref{fl_moju}, i.e.\
	\begin{equation}
		\label{general_case}
		g_{\mu\nu}\,,\,\, \Phi \,,\,\, \text{etc.}
	\end{equation}
	\item[ii) ] An arbitrary parametrization $\mathfrak{P}$, see Definition \ref{def:parametrization}, where we shall add a superscript $\mathfrak{P}$ to the metric tensor and a subscript $\mathfrak{P}$ to all other objects as
	\begin{equation}
		\label{parametrization_case}
		g^{\mathfrak{P}}_{\mu\nu}\,,\,\, \Phi_{\mathfrak{P}}\,,\,\, \mathcal{A}_{\mathfrak{P}} \equiv \mathcal{A}_{\mathfrak{P}}(\Phi_\mathfrak{P})\,,\,\, \text{etc.}
	\end{equation}
	\item[iii) ]
	The invariant case determined by a parametrization $\mathfrak{P}$. There we use an invariant metric \eqref{invariant_metric} and other invariants
	\begin{equation}
		\label{invariant_case}
		\hat{g}^{(\mathfrak{P})}_{\mu\nu} \,,\,\, \mathcal{I}^{(\mathfrak{P})}(\Phi) \,,\,\,\text{etc.}
	\end{equation}
	Here $\mathfrak{P}$ as a superscript in parentheses emphasizes that the quantity under consideration is determined by the parametrization $\mathfrak{P}$ but does not have to be evaluated in that parametrization. It could be calculated in any other parametrization or instead considered in the generic case. What it means to be determined by a parametrization $\mathfrak{P}$ will be clarified in the following pages.
\end{itemize}

\noindent There are six possibilities to fix two out of the four arbitrary functionals $\left\lbrace \mathcal{A},\,\mathcal{B},\,\mathcal{V},\,\alpha  \right\rbrace$, i.e.\ to choose a parametrization. For four possibilities out of the six a quick glimpse on \eqref{fl_fn_teisendused} reveals that also one invariant gains a fixed functional form. Namely
	\begin{itemize}
		\item[i) ] $\mathcal{A}$ and $\alpha$ are fixed:\,\,
		$\mathcal{I}_1(\Phi_\mathfrak{P}) \equiv \frac{e^{2\alpha_\mathfrak{P}}}{\mathcal{A}_\mathfrak{P}}$,
		\item[ii) ] $\mathcal{A}$ and $\mathcal{V}$ are fixed:\,\,
		$\mathcal{I}_2(\Phi_\mathfrak{P})\equiv \frac{\mathcal{V}_\mathfrak{P}}{\mathcal{A}_\mathfrak{P}^2}$,
		\item[iii) ] $\mathcal{A}$ and $\mathcal{B}$ are fixed:\,\,
		$\mathcal{I}_3(\Phi_\mathfrak{P}) \equiv \pm \bigints \sqrt{\frac{2\mathcal{A}_\mathfrak{P}\mathcal{B}_\mathfrak{P} + 3\left(\mathcal{A}_\mathfrak{P}^\prime\right)^2}{4\mathcal{A}_\mathfrak{P}^2}} \delta \Phi_\mathfrak{P}$,
		\item[iv) ] $\mathcal{V}$ and $\alpha$ are fixed:\,\,
		$\mathcal{I}_4(\Phi_\mathfrak{P}) \equiv \frac{\mathcal{V}_\mathfrak{P}}{e^{4\alpha_\mathfrak{P}}}$.
	\end{itemize}
The case where $\mathcal{B}$ and $\mathcal{\alpha}$ (analogically $\mathcal{B}$ and $\mathcal{V}$) have a fixed functional form is more complicated: the corresponding invariant (if it exists) depends on the exact functional form of $\mathcal{B}$ and $\alpha$ and is not the same for all possible choices. For an example see JF BEPS in \cite{JKSV_2}.

\begin{lem}\label{lem:gamma}
	Let us assume that in a parametrization $\mathfrak{P}$ an invariant $\mathcal{I}_{fix} (\Phi_{\mathfrak{P}})$ has gained a fixed functional form. If $\mathcal{I}_{fix}$ is a nonconstant local functional then there exists a functional $\mathcal{K}^{(\mathfrak{P})}(\Phi)$ which in the parametrization $\mathfrak{P}$ is equal to $1$ and in the generic case transforms as $\mathcal{A}(\Phi)$, i.e.\ according to \eqref{A_transformation}.
\end{lem}

\noindent Note that by writing $\mathcal{K}^{(\mathfrak{P})}(\Phi)$ we abuse the notation \eqref{invariant_case} since it is not an invariant but we make an exception because it is determined by a parametrization $\mathfrak{P}$ and yet does not have to be evaluated in $\mathfrak{P}$.

\begin{proof}
	Let us consider a parametrization $\mathfrak{P}$. If the premiss is fulfilled then $\mathcal{I}_{fix}(\Phi_\mathfrak{P}) = h(\Phi_{\mathfrak{P}})$ is a known nonconstant local functional. We invert the latter to obtain a possibly multivalued relation $\Phi_\mathfrak{P} = h^{-1}(\mathcal{I}_{fix})$. In the current paper we do not consider the consequences of multivaluedness. According to the Corollary~\ref{cor1} a functional of an invariant is also an invariant and therefore in the parametrization $\mathfrak{P}$ it is meaningful to write $\Phi_{\mathfrak{P}} = \mathcal{I}^{(\mathfrak{P})}$ where $\mathcal{I}^{(\mathfrak{P})} \equiv h^{-1}(\mathcal{I}_{fix})$. Note that $\mathcal{I}^{(\mathfrak{P})}$ is determined by the parametrization $\mathfrak{P}$ but otherwise is an ordinary invariant. In a sense $\Phi_{\mathfrak{P}} = \mathcal{I}^{(\mathfrak{P})}(\Phi)$ relates the scalar field $\Phi_\mathfrak{P}$ to a generic scalar field $\Phi$.
	
	According to the Definition~\ref{def:parametrization} two out of the four arbitrary local functionals $\left\lbrace \mathcal{A},\,\mathcal{B},\,\mathcal{V},\,\alpha  \right\rbrace$ of $\Phi$ have gained a fixed functional form. Therefore one must be either $\mathcal{A}$, $\mathcal{V}$ or $\alpha$. 
	
	First let us consider the case where the functional form of $\mathcal{A}_\mathfrak{P}(\Phi_\mathfrak{P}) \equiv \left.\mathcal{A}(\Phi)\right|_\mathfrak{P}$ is fixed. We make use of the result $\Phi_{\mathfrak{P}} = \mathcal{I}^{(\mathfrak{P})}$ and replace the argument of $\mathcal{A}_{\mathfrak{P}}(\Phi_\mathfrak{P})$ as $\mathcal{A}_{\mathfrak{P}} \equiv \mathcal{A}_{\mathfrak{P}}(\mathcal{I}^{(\mathfrak{P})})$. The Corollary~\ref{cor1} states that the obtained quantity is an invariant. By making use of the notation introduced in \eqref{invariant_case} we write $\mathcal{A}^{(\mathfrak{P})}(\Phi) \equiv \mathcal{A}_{\mathfrak{P}}(\mathcal{I}^{(\mathfrak{P})}(\Phi))$ to denote an invariant with the property $\left.\mathcal{A}^{(\mathfrak{P})}( \Phi )\right|_\mathfrak{P} = \mathcal{A}_\mathfrak{P}(\Phi_\mathfrak{P})$. Hence $\mathcal{A}^{(\mathfrak{P})}(\Phi)$ is an invariant which is determined by the parametrization $\mathfrak{P}$ but can be considered in whatever case. In the generic case the quotient
	\begin{equation}
	\label{one_1}
	\mathcal{K}^{(\mathfrak{P})}(\Phi) \equiv \frac{\mathcal{A}(\Phi)}{\mathcal{A}^{(\mathfrak{P})}(\Phi)}
	\end{equation} 
	is a local functional of $\Phi$ that transforms as $\mathcal{A}(\Phi)$ and in the parametrization $\mathfrak{P}$ we obtain that $\left.\mathcal{K}^{(\mathfrak{P})}\right|_\mathfrak{P} = 1$. 
	
	The proof in the case when the functional form of either $\mathcal{V}$ or $\alpha$ is fixed proceeds analogically.
\end{proof}

\noindent Note that formally each functional that transforms as $\mathcal{A}$, i.e.\ according to \eqref{A_transformation}, can be written as a product of $\mathcal{A}$ and some invariant, e.g.\ $e^{2\alpha} \equiv \mathcal{A}\,\mathcal{I}_1$.

\begin{ex}\label{example1} 
	We consider a parametrization $\mathfrak{P}$ where $\mathcal{A}_\mathfrak{P} \equiv \Phi_\mathfrak{P}$ and $e^{2\alpha_\mathfrak{P}} \equiv 1 + \lambda\Phi_\mathfrak{P}$. Here $\lambda$ is some constant parameter. The scalar field $\Phi_\mathfrak{P}$ can be expressed as a local functional of the fixed invariant $\mathcal{I}_1$ as follows
	\begin{equation}
	\Phi_\mathfrak{P} = \frac{1}{\mathcal{I}_1 - \lambda} \equiv \mathcal{I}^{(\mathfrak{P})} \,.
	\end{equation}
	Hence $\mathcal{A}^{(\mathfrak{P})} \equiv \mathcal{I}^{(\mathfrak{P})}$ and the quotient
	\begin{equation}
	\label{quotient}
	\mathcal{K}^{(\mathfrak{P})}(\Phi) \equiv \frac{\mathcal{A}(\Phi)}{\mathcal{A}^{(\mathfrak{P})} ( \Phi )} \equiv \left(\mathcal{I}_1(\Phi) - \lambda\right) \mathcal{A}(\Phi) \,
	\end{equation}
	has the demanded properties. A direct calculation shows that if we use an analogous procedure but consider $e^{2\alpha_\mathfrak{P}}$ instead of $\mathcal{A}_\mathfrak{P}$ then we get the same result.
	
	The result for JF BDBW parametrization \eqref{Jordan_frame} is obtained by fixing $\lambda \equiv 0$. In that case the result \eqref{quotient} reduces to 
	\begin{equation}
	\mathcal{K}^{(\mathfrak{J})} \equiv \mathcal{A}\,\mathcal{I}_1 \equiv e^{2\alpha}
	\end{equation}
	which in JF BDBW parametrization is indeed equal to one and in the generic case transforms as $\mathcal{A}$. For EF canonical parametrization \eqref{Einstein_frame} $\mathcal{K}^{(\mathfrak{E})} \equiv \mathcal{A}$.
\end{ex}

The relation $\Phi_\mathfrak{P} = \mathcal{I}^{(\mathfrak{P})}$ in the parametrization $\mathfrak{P}$, obtained in the proof of the Lemma~\ref{lem:gamma}, introduces an another object which in the parametri\-zation $\mathfrak{P}$ is equal to one but has a specific transformation property. Namely in the parametrization $\mathfrak{P}$
\begin{equation}
	\label{one_2}
		1 = \frac{\delta \Phi_\mathfrak{P}}{\delta \Phi_\mathfrak{P}} = \frac{\delta \mathcal{I}^{(\mathfrak{P})}}{ \delta \Phi_\mathfrak{P} } \,.
\end{equation}
In the generic case $\bar{\mathcal{I}}^{(\mathfrak{P})\,\prime} = \bar{f}^\prime \mathcal{I}^{(\mathfrak{P})\,\prime}$.

\begin{thm}\label{theorem2}
	If, due to specifying the parametrization to be $\mathfrak{P}$, an invariant gains a fixed nonconstant functional form then there exists an `invariant pair'
	\begin{equation}
	\label{invariant_pair}
	\left( \hat{g}^{(\mathfrak{P})}_{\mu\nu} ,\, \mathcal{I}^{(\mathfrak{P})} \right) \,
	\end{equation}
	which in the parametrization $\mathfrak{P}$ functionally coincides with the pair $\left( g^{\mathfrak{P}}_{\mu\nu},\Phi_\mathfrak{P} \right)$.
\end{thm}

\begin{proof}
	Let us consider a parametrization $\mathfrak{P}$. If the premiss holds then the Lemma~\ref{lem:gamma} proposes the existence of the functional $\mathcal{K}^{(\mathfrak{P})}(\Phi)$ which has the properties: $ \mathcal{K}^{(\mathfrak{P})} = e^{-2\bar{\gamma}} \bar{\mathcal{K}}^{(\mathfrak{P})}$ and $\left.\mathcal{K}^{(\mathfrak{P})}(\Phi)\right|_\mathfrak{P} = 1 $. Hence $\hat{g}^{(\mathfrak{P})}_{\mu\nu} \equiv \mathcal{K}^{(\mathfrak{P})}g_{\mu\nu}$ is an invariant metric \eqref{invariant_metric} and in the parametrization $\mathfrak{P}$
	\begin{equation}
	\left.\hat{g}^{(\mathfrak{P})}_{\mu\nu} \right|_\mathfrak{P} = g^{\mathfrak{P}}_{\mu\nu} \,.
	\end{equation}
	In the same spirit $\left.\mathcal{I}^{(\mathfrak{P})} \right|_\mathfrak{P} = \Phi_\mathfrak{P}$ holds by the definition introduced in the proof of the Lemma~\ref{lem:gamma}.
\end{proof}

\begin{ex}
	In JF BDBW parametrization \eqref{Jordan_frame}
	\begin{equation}
	\label{invariant_pair_in_JF_BDBW}
	\left.\left( \hat{g}^{(\mathfrak{J})}_{\mu\nu},\,\frac{1}{\mathcal{I}_1} \right)\right|_{\mathfrak{J}} \equiv \left.\left( e^{2\alpha}g_{\mu\nu},\,\frac{1}{\mathcal{I}_1} \right)\right|_{\mathfrak{J}} = (g^{\mathfrak{J}}_{\mu\nu}, \Psi)\,.
	\end{equation}
	In EF canonical parametrization \eqref{Einstein_frame} 
	\begin{equation}
	\label{invariant_pair_in_EF_canonical}
	\left.\left(\hat{g}^{(\mathfrak{E})}_{\mu\nu},\,\pm\,\mathcal{I}_3\right)\right|_{\mathfrak{E}} \equiv \left.\left(\mathcal{A}g_{\mu\nu},\,\pm\,\mathcal{I}_3 \right)\right|_{\mathfrak{E}} = (g^{\mathfrak{E}}_{\mu\nu},\,\varphi) \,.
	\end{equation}
\end{ex}

Let us take the metric tensor from the invariant pair \eqref{invariant_pair}, determined by some parametrization $\mathfrak{P}$, and calculate the invariant Ricci scalar $\hat{R}^{(\mathfrak{P})}$ for that metric tensor. In the parametrization $\mathfrak{P}$ the invariant Ricci scalar $\hat{R}^{(\mathfrak{P})}$ functionally coincides with the Ricci scalar $R_{\mathfrak{P}}$ that is calculated using the metric tensor $g^{\mathfrak{P}}_{\mu\nu}$.

\begin{ex}
	Let us consider JF BDBW parametrization \eqref{Jordan_frame} that determines the invariant pair \eqref{invariant_pair_in_JF_BDBW}. One can show that \cite{Nakahara}
	\begin{equation}
	\label{Ricci_scalar_in_Jordan_frame}
	e^{2\alpha}\hat{R}^{(\mathfrak{J})} = R - 6 g^{\mu\nu} \nabla_\mu \alpha \nabla_\nu \alpha - 6g^{\mu\nu}\nabla_\mu \nabla_\nu \alpha \,
	\end{equation}
	where the r.h.s.\ is calculated for the generic case \eqref{general_case}. Restricting Eq.~\eqref{Ricci_scalar_in_Jordan_frame} to the JF BDBW parametrization \eqref{Jordan_frame} ($\alpha\equiv 0$) gives us the equality
	\begin{equation}
	\left.\hat{R}^{(\mathfrak{J})}\right|_{\mathfrak{J}} = R_{\mathfrak{J}} \,.
	\end{equation}
	The result \eqref{Ricci_scalar_in_Jordan_frame} resembles the transformation of the Ricci scalar under the Weyl rescaling. Here, in the spirit of the Remark~\ref{rem3}, we introduce additional terms to cancel the effect of the conformal transformation on the Ricci scalar. 
\end{ex}

\section{The relation between the generic case and a chosen parametrization revisited. The translation rules.}\label{translation_rules}

Let us consider an invariant pair \eqref{invariant_pair} determined by a parametrization $\mathfrak{P}$.
If one rewrites the action functional \eqref{fl_moju} using the components of the invariant pair \eqref{invariant_pair} as the dynamical variables then four invariants, which we shall denote as $\left\lbrace \mathcal{I}^{(\mathfrak{P})}_\mathcal{A},\, \mathcal{I}^{(\mathfrak{P})}_\mathcal{B},\, \mathcal{I}^{(\mathfrak{P})}_\mathcal{V},\, \mathcal{I}^{(\mathfrak{P})}_\alpha \right\rbrace$, appear into the positions of the four arbitrary local functionals $\left\lbrace \mathcal{A},\, \mathcal{B},\, \mathcal{V},\, \alpha \right\rbrace$.

Such a claim can be reasoned as follows. Using the invariant metric $\hat{g}^{(\mathfrak{P})}_{\mu\nu}$ to calculate geometrical quantities guarantees that the latter are invariant under the transformations \eqref{conformal_transformation}-\eqref{field_redefinition}. In the same spirit the kinetic term for $\mathcal{I}^{(\mathfrak{P})}$ is invariant as well. Therefore there is no mixing of the additive terms in the action functional $S\left[\hat{g}^{(\mathfrak{P})}_{\mu\nu},\, \mathcal{I}^{(\mathfrak{P})},\, \chi \right]$ under the transformations \eqref{conformal_transformation}-\eqref{field_redefinition}. We conclude that for such an action functional each additive term must be an invariant by itself because we have assumed the action functional \eqref{fl_moju} to be invariant. Each of the four arbitrary local functionals $\left\lbrace \mathcal{A},\, \mathcal{B},\, \mathcal{V},\, \alpha \right\rbrace$ multiplies an object which after rewriting is replaced by an invariant. Therefore during the rewriting process the four arbitrary local functionals must be replaced by invariants as well.

\begin{ex}
	First let us consider JF BDBW parametrization \eqref{Jordan_frame}. Rewriting the action functional \eqref{fl_moju} in terms of the invariant pair \eqref{invariant_pair_in_JF_BDBW} reads
	\begin{align}
	\nonumber
	S = \frac{1}{2\kappa^2} \int_{V_4} d^4x\sqrt{-\hat{g}^{(\mathfrak{J})}} \Bigg\lbrace \frac{1}{\mathcal{I}_1} \hat{R}^{(\mathfrak{J})} &- \mathcal{I}_1\frac{1}{2}\left( \frac{1}{\mathcal{I}_5}\negmedspace - \negmedspace 3 \right) \hat{g}^{(\mathfrak{J})\mu\nu} \hat{\nabla}^{(\mathfrak{J})}_\mu  \frac{1}{\mathcal{I}_1}  \hat{\nabla}^{(\mathfrak{J})}_\nu \frac{1}{\mathcal{I}_1} \\
	\label{action_in_terms_of_invariants_in_JF_BDBW}
	&-2\ell^{-2} \mathcal{I}_4 \Bigg\rbrace + S_m\left[ \hat{g}^{(\mathfrak{J})}_{\mu\nu},\,\chi \right] \,.
	\end{align}
	Here we have made use of the definitions \eqref{I_1}-\eqref{I_3} and \eqref{I_4} and of the result \eqref{Ricci_scalar_in_Jordan_frame}. Hence $\mathcal{I}^{(\mathfrak{J})}_\mathcal{A} = \frac{1}{\mathcal{I}_1}$, $\mathcal{I}^{(\mathfrak{J})}_\mathcal{B} = \mathcal{I}_1\frac{1}{2}\left( \frac{1}{\mathcal{I}_5}\negmedspace - \negmedspace 3 \right)$, $\mathcal{I}^{(\mathfrak{J})}_\mathcal{V} = \mathcal{I}_4$ and $\mathcal{I}^{(\mathfrak{J})}_\alpha = 0$.
	
	Second let us consider EF canonical parametrization \eqref{Einstein_frame} and the corresponding invariant pair \eqref{invariant_pair_in_EF_canonical}. One can rewrite the action functional \eqref{fl_moju} as
	\begin{align}
	\nonumber
	S  = \frac{1}{2\kappa^2} \int_{V_4} d^4x \sqrt{-\hat{g}^{(\mathfrak{E})}} \Bigg\lbrace \hat{R}^{(\mathfrak{E})} -& 2 \hat{g}^{(\mathfrak{E})\mu\nu} \hat{\nabla}^{(\mathfrak{E})}_\mu \mathcal{I}_3 \hat{\nabla}^{(\mathfrak{E})}_\nu\mathcal{I}_3 - 2\ell^{-2}\mathcal{I}_2 \Bigg\rbrace\\
	\label{action_functional_in_terms_of_invariants_for_EF_can}
	&+ S_m\left[ \mathcal{I}_1 \hat{g}^{(\mathfrak{E})}_{\mu\nu},\,\chi \right] \,.
	\end{align}
	In this example $\mathcal{I}^{(\mathfrak{E})}_\mathcal{A} = 1$, $\mathcal{I}^{(\mathfrak{E})}_\mathcal{B} = 2$, $\mathcal{I}^{(\mathfrak{E})}_\mathcal{V} = \mathcal{I}_2$ and $\mathcal{I}^{(\mathfrak{E})}_\alpha = \frac{1}{2}\ln \mathcal{I}_1$.
\end{ex}

Rewriting the action functional \eqref{fl_moju} in terms of an invariant pair \eqref{invariant_pair} retains the generality of the theory up to some minor details that we shall not discuss in the current paper.

\begin{cor}
	Let us rewrite the general action functional $S = S\left[ g_{\mu\nu},\Phi,\chi \right]$, defined by \eqref{fl_moju}, using the components of an invariant pair $\left( \hat{g}^{(\mathfrak{P})}_{\mu\nu} ,\, \mathcal{I}^{(\mathfrak{P})} \right)$, determined by a parametrization $\mathfrak{P}$, as dynamical variables. We end up with an action functional $S = S\left[ \hat{g}^{(\mathfrak{P})}_{\mu\nu},\mathcal{I}^{ (\mathfrak{P}) },\chi \right]$ involving a boundary term which we shall neglect. Let us focus upon the action functional in terms of the invariants. If we specify the theory by fixing the parametrization to be $\mathfrak{P}$ then each invariant quantity is mapped to the corresponding noninvariant quantity in the parametrization $\mathfrak{P}$ as follows
	\begin{equation}
	\label{mapping_rules}
	\begin{tabular}{rclcrcl}
	$\hat{g}^{(\mathfrak{P})}_{\mu\nu} $ &$\mapsto$& $g^{\mathfrak{P}}_{\mu\nu}$\qquad &, \qquad &$\mathcal{I}^{(\mathfrak{P})}_\mathcal{A}$ &$\mapsto$& $\mathcal{A}_\mathfrak{P}$\,,  \\
	$\sqrt{-\hat{g}^{(\mathfrak{P})}} $&$\mapsto$& $\sqrt{-g^{\mathfrak{P}}}$\qquad &,\qquad & $ \mathcal{I}^{(\mathfrak{P})}_\mathcal{B} $ &$\mapsto$& $\mathcal{B}_\mathfrak{P}$\,, \\
	$ \hat{R}^{(\mathfrak{P})}$&$\mapsto$& $R_{\mathfrak{P}}$ \qquad&,\qquad & $\mathcal{I}^{(\mathfrak{P})}_\mathcal{V}$ &$\mapsto$& $\mathcal{V}_\mathfrak{P}$\,, \\
	$\hat{\nabla}^{(\mathfrak{P})}_\mu $ &$\mapsto$& $\nabla^{\mathfrak{P}}_\mu$ \qquad&,\qquad & $\mathcal{I}^{(\mathfrak{P})}_\alpha $ &$\mapsto$& $\alpha_\mathfrak{P}$\,, \\
	$ \mathcal{I}^{(\mathfrak{P})} $ &$\mapsto$& $\Phi_\mathfrak{P}$\qquad&. & &&
	\end{tabular}
	\end{equation}
\end{cor}

\begin{ex}
	First let us consider JF BDBW parametrization \eqref{Jordan_frame}. The action functional \eqref{fl_moju} rewritten in terms of the invariant pair \eqref{invariant_pair_in_JF_BDBW} is given by \eqref{action_in_terms_of_invariants_in_JF_BDBW}. A straightforward calculation shows that fixing the parametrization to be JF BDBW parametrization implies
	\begin{align}
	\left.\frac{1}{\mathcal{I}_1}\right|_{\mathfrak{J}}= \Psi \equiv \mathcal{A}_{\mathfrak{J}} \,,\quad &\left.\mathcal{I}_1\frac{1}{2}\left( \frac{1}{\mathcal{I}_5}\negmedspace - \negmedspace 3 \right)\right|_{\mathfrak{J}} = \frac{\omega(\Psi)}{\Psi} \equiv \mathcal{B}_{\mathfrak{J}} \,, \\ \left.\mathcal{I}_4\right|_{\mathfrak{J}} = \mathcal{V}_{\mathfrak{J}}(\Psi)\,,\quad &\mathcal{I}^{({\mathfrak{J}})}_\alpha = 0 = \alpha_{\mathfrak{J}} \,.
	\end{align}

	Second let us consider EF canonical parametrization \eqref{Einstein_frame}. The invariant pair \eqref{invariant_pair_in_EF_canonical} gives rise to the action functional \eqref{action_functional_in_terms_of_invariants_for_EF_can}. A direct calculation shows that 
	\begin{equation}
	1\equiv \mathcal{A}_{\mathfrak{E}} \,,\quad 2 \equiv \mathcal{B}_{\mathfrak{E}} \,,\quad \left.\mathcal{I}_2\right|_{\mathfrak{E}} = \mathcal{V}_{\mathfrak{E}}(\varphi) \,,\quad \left. \frac{1}{2}\ln \mathcal{I}_1 \right|_{\mathfrak{E}} = \alpha_{\mathfrak{E}}(\varphi) \,.
	\end{equation}
\end{ex}

\begin{rem}\label{cor4}
	Let us consider the case where we have two action functionals $S_1$ and $S_2$. The action $S_1 \equiv S_1\left[g^{\mathfrak{P}}_{\mu\nu},\,\Phi_\mathfrak{P},\,\chi\right]$ is obtained from \eqref{fl_moju} by fixing the parametrization to be $\mathfrak{P}$ and $S_2 \equiv S_2\left[ \hat{g}^{(\mathfrak{P})}_{\mu\nu},\, \mathcal{I}^{(\mathfrak{P})},\, \chi \right]$ is obtained by rewriting the action functional \eqref{fl_moju} in terms of the invariant pair \eqref{invariant_pair} that is determined by $\mathfrak{P}$. Suppose that we are given an action functional $S_3$ and we know that $S_3$ is either $S_1$ or $S_2$. Due to the one to one correspondence \eqref{mapping_rules} we cannot determine whether $S_3$ is $S_1$ or $S_2$ without a priori knowing how the quantities contained in $S_3$ transform, i.e.\ whether the transformation of the quantities obey Eqs.~\eqref{conformal_transformation}-\eqref{alpha_transformation} or the rules described after Eq.~\eqref{I_3}. Therefore without a priori given transformation rules the action functionals $S_1$ and $S_2$ cannot be distinguished.
\end{rem}
Let us point out that the redefinition of the scalar field can be seen as choosing a different invariant to be the dynamical variable.
\begin{ex}
	Lets us consider JF BDBW parametrization \eqref{Jordan_frame} scalar field $\Psi$ as a local functional of the EF canonical parametrization \eqref{Einstein_frame} scalar field $\varphi$. By comparing the invariant pairs \eqref{invariant_pair_in_JF_BDBW} and \eqref{invariant_pair_in_EF_canonical} we obtain that this corresponds to
	\begin{equation}
	\frac{1}{\mathcal{I}_1} \equiv \frac{1}{\mathcal{I}_1(\mathcal{I}_3)} \,.
	\end{equation}
	Hence
	\begin{equation}
	\label{EF_vs_JF_via_invariants}
	\left( \frac{\delta \Psi}{\delta \varphi} \right)^2 = \left(\frac{\delta \frac{1}{\mathcal{I}_1}}{\delta \mathcal{I}_3} \right)^2 = \left( \frac{\mathcal{I}_1^\prime}{\mathcal{I}_1^2 \, \mathcal{I}_3^\prime} \right)^2 = \frac{4 \mathcal{I}_5}{\mathcal{I}_1^2} \,
	\end{equation}
	where we made use of the definition \eqref{I_4}.
	If the result is evaluated in EF canonical parametrization \eqref{Einstein_frame} then it agrees with Eq.~\eqref{EF_to_JF}. If Eq.~\eqref{EF_vs_JF_via_invariants} is evaluated in JF BDBW parametrization \eqref{Jordan_frame} then it agrees with \eqref{JF_to_EF}.
\end{ex}

The one to one correspondence \eqref{mapping_rules} gives rise to the `translation rules' that were first implicitly used in Ref.~\cite{JKSV_1} and more thoroughly studied in Ref.~\cite{JKSV_2}. The translation rules can be used to rewrite the results obtained in some parametrization $\mathfrak{P}$ as the results of the generic case described by the action functional \eqref{fl_moju}. The key idea can be phrased as follows.
\begin{itemize}
	\item[i) ] Calculate the invariant pair \eqref{invariant_pair} determined by a parametrization $\mathfrak{P}$.
	\item[ii) ] Rewrite the action functional \eqref{fl_moju} in terms of the obtained invariant pair and determine the l.h.s.\ of the correspondence \eqref{mapping_rules}.
	\item[iii) ] Replace each quantity in the parametrization $\mathfrak{P}$ by the corresponding invariant, i.e.\ use the mapping \eqref{mapping_rules} backwards.
	\item[iv) ] Evaluate the obtained invariant quantities in terms of the four arbitrary local functionals $\left\lbrace \mathcal{A},\,\mathcal{B},\,\mathcal{V},\,\alpha \right\rbrace$ and use a generic metric tensor $g_{\mu\nu}$ and a generic scalar field $\Phi$ as dynamical variables.
\end{itemize}

Instead of following the second rule of the aforementioned prescription one can use the transformations \eqref{fl_fn_teisendused} to obtain the invariants that correspond to the four local functionals $\left\lbrace \mathcal{A}_\mathfrak{P},\,\mathcal{B}_\mathfrak{P},\, \mathcal{V}_\mathfrak{P},\, \alpha_\mathfrak{P} \right\rbrace$ in a parametrization $\mathfrak{P}$.

Namely, let us consider the quantities of the invariant case to be formally the ``barred" ones. The definition of the invariant metric in the invariant pair \eqref{invariant_pair}, i.e.\ $\hat{g}^{(\mathfrak{P})}_{\mu\nu} \equiv \mathcal{K}^{(\mathfrak{P})} g_{\mu\nu}$ can be seen as a Weyl rescaling of the metric tensor \eqref{conformal_transformation} where $e^{2\bar{\gamma}(\mathcal{I}^{(\mathfrak{P})})} = \left( \mathcal{K}^{(\mathfrak{P})}(\Phi( \mathcal{I}^{(\mathfrak{P})} )) \right)^{-1}$. The crucial point is that for generic case
\begin{equation}
	\left(\bar{\mathcal{K}}^{(\mathfrak{P})} \right)^{-1} = e^{-2\bar{\gamma}} \left( \mathcal{K}^{(\mathfrak{P})} \right)^{-1} \,.
\end{equation} 
Therefore using $e^{2\bar{\gamma}} = \left( \mathcal{K}^{(\mathfrak{P})} \right)^{-1}$ for performing the transformations \eqref{fl_fn_teisendused} actually, in the spirit of the Remark \ref{rem3}, introduces extra terms with suitable transformation properties to cancel the effect of the Weyl rescaling on the arbitrary local functionals $\left\lbrace \mathcal{A},\,\mathcal{B},\,\mathcal{V},\,\alpha  \right\rbrace$. Analogically $\bar{f}^\prime = \left( \mathcal{I}^{(\mathfrak{P})\,\prime} \right)^{-1}$.

There are noninvariant objects that in a parametrization $\mathfrak{P}$ are equal to one, e.g.\ \eqref{one_1} and \eqref{one_2} and various combinations of these. Therefore the translation rules cannot directly determine the transformation properties and hence can work fluently only in the case of invariant quantities. There are indirect ways to obtain the transformation properties as well, e.g.\ comparing the results calculated from different parametrizations.


\subsection*{Acknowledgment}

This work was supported by the Estonian Science Foundation Grant No.~8837, by the Estonian Research Council Grant No.~IUT02-27 and by the European Union through the European Regional Development Fund (Project No. 3.2.0101.11-0029). The author would like to thank Piret Kuusk, Laur J\"arv and Margus Saal for fruitful discussions.


\end{document}